\documentclass[11pt]{article}
\usepackage{amsmath}
\usepackage{amssymb}
\usepackage{fullpage}
\usepackage{amsthm}
\usepackage{epsf}
\usepackage{graphicx}
\usepackage{color}
\usepackage{enumerate}
\RequirePackage{hyperref}

\def\blackslug{\rule{2mm}{2mm}}
\def\QED{\hfill\blackslug}

\newcommand{\ignore}[1]{}

\theoremstyle{plain}

\newtheorem{theorem}{Theorem}
\newtheorem{lemma}[theorem]{Lemma}

\newtheorem{corollary}[theorem]{Corollary}
\newtheorem{definition}{Definition}

\theoremstyle{definition}
\newtheorem{remark}{Remark}

\renewenvironment{proof}{\noindent {\bf \em Proof\/}.\enspace}
{\hfill $\blacksquare{}$ \vspace{12pt}}

\def\pr{\mathop{\bf Pr}\limits}

\renewcommand{\leq}{\leqslant}
\renewcommand{\geq}{\geqslant}
\renewcommand{\le}{\leqslant}
\renewcommand{\ge}{\geqslant}
\newcommand{\eps}{\varepsilon}
\renewcommand{\epsilon}{\varepsilon}

\newcommand{\vnote}[1]{}
\newcommand{\anote}[1]{}

\newcommand{\p}{\rho}
\newcommand{\F}{\mathbb{F}}

\newcommand{\angles}[1]{\langle #1 \rangle}

\renewcommand{\a}{\alpha}
\renewcommand{\b}{\beta}
\newcommand{\g}{\gamma}

\newcommand{\cout}{C_{\rm out}}
\newcommand{\cin}{C_{\rm in}}
\newcommand{\myvec}[1]{\mathbf{#1}}
\newcommand{\vg}{\myvec{G}}

\newcommand{\vu}{\myvec{u}}
\newcommand{\vm}{\myvec{m}}
\newcommand{\vx}{\myvec{x}}

\newcommand{\vy}{\myvec{y}}

\newcommand{\vd}{\myvec{d}}
\newcommand{\vc}{\myvec{c}}

\newcommand{\ve}{\myvec{e}}

\newcommand{\vol}[1]{\mathrm{Vol}_{#1}}

\newcommand{\wt}{\textsc{wt}}
\begin{document}

\title{{\bf Two Theorems in List Decoding}\footnote{Research
supported by NSF CAREER Award CCF-0844796.}}

\author{
{\sc Atri Rudra}\and {\sc Steve Uurtamo}}

\date{Department of Computer Science and Engineering,\\
University at Buffalo, The State University of New York,\\
Buffalo, NY, 14620.\\  {\tt \{atri,uurtamo\}@buffalo.edu}}

\maketitle

\setcounter{page}{0}
\thispagestyle{empty}

\begin{abstract}

\noindent
We prove the following results concerning the list decoding of error-correcting codes:
\begin{enumerate}
\item We show that for \textit{any} code with a relative distance of $\delta$
(over a large enough alphabet), the
following result holds for \textit{random errors}: With high probability,
for a $\rho\le \delta -\eps$ fraction of random errors (for any $\eps>0$), 
the received word will have only the transmitted codeword in a Hamming ball of 
radius $\rho$ around it. Thus, for random errors, one can correct \textit{twice} the
number of errors uniquely correctable from worst-case errors for any code. 
A variant of our
result also gives a simple algorithm to decode Reed-Solomon codes from
random errors that, to the best of our knowledge, runs faster than known
algorithms for certain ranges of parameters.
\item We show that concatenated codes can achieve the list decoding capacity
for \textit{erasures}. A similar result for worst-case
errors was proven by Guruswami and Rudra (SODA 08), although their result does 
not directly imply our result. Our results show that a subset of 
the random ensemble of codes considered by Guruswami and Rudra also achieve
the list decoding capacity for erasures.
\end{enumerate}
Our proofs employ simple counting and probabilistic arguments.
\end{abstract}

\newpage

\section{Introduction}

List decoding is a relaxation of the traditional unique decoding paradigm,
where one is allowed to output a list of codewords that are close to the
received word. This relaxation allows for designing list decoding algorithms
that can recover from scenarios where almost all of the redundancy could have
been corrupted~\cite{sudan,GS98,PV-focs05,GR-capacity}. In particular, one can
design binary codes from which one can recover from a $1/2-\eps$ fraction of errors.
This fact has lead to many surprising applications in complexity theory-- see
e.g. the survey by Sudan~\cite{sudan-sigact} and Guruswami's thesis~\cite[Chap. 12]{G-thesis}.

The results mentioned above mostly deal with worst-case errors, where 
the channel is considered to be an adversary that can corrupt any arbitrary
fraction of symbols (with an upper bound on the maximum fraction of such errors).
In this work, we deal with random and erasure noise models,
 which are weaker than the worst-case
errors model, and which also have interesting applications in complexity theory.


\subsection{Random Errors} 
It is well-known that for worst-case errors, one
cannot uniquely recover the transmitted codeword if the total number of errors
exceeds half the distance. (We refer the reader to Section~\ref{sec:prelims}  for definitions
related to codes.) 
List decoding circumvents this by allowing the decoder to output
multiple nearby codewords. In situations where the decoder has access to some
side information, one can prune the output list to obtain the transmitted 
codeword. In fact, most of the applications of list decoding in complexity
theory crucially use side information. However, a natural question to ask
is what one can do in situations where there is no side information (this is
not an uncommon assumption in the traditional point-to-point communication
model).

In such a scenario, it makes sense to look at a weaker random noise model and try
to argue that the pathological cases that prevent us from decoding a code
with relative distance $\delta$ from more than $\delta/2$ fraction of errors
are rarely encountered.

Before we move on, we digress a bit to establish our notion of random errors.
In our somewhat non-standard model, we assume that the adversary can pick the
location of the $\p$ fraction of error \textit{positions} but that the
errors themselves are random.
For the binary case, this model coincides with worst-case errors, so in 
this work, we consider alphabet size $q\ge 3$. We believe that this is a nice
intermediary to the worst-case noise model and the more popular models of random
noise, where errors are \textit{independent} across different symbols. Indeed,
a result with high probability in our random noise model (for roughly
$\p$ errors)
immediately implies a similar result for a more benign random noise model such as
the $q$-ary symmetric noise channel with cross-over probability $\p$.\footnote{In
this model, every transmitted symbol   remains untouched with probability $1-\p$
and is mapped to the other $q-1$ possible symbols with probability $\p/(q-1)$. Finally, the noise acts independently on each symbol.} For the rest of the
paper, when we say random errors, we will be referring to the stronger random
noise model above.

\paragraph{Related Work.} The intuition that pathological worst-case errors are rare has been 
formalized for certain families of codes. For example, McEliece showed that
for Reed-Solomon codes with distance $\delta$, with high probability, for a
fraction $\p\le \delta-\eps$ of random errors, the output list size is one~\cite{mceliece}.\footnote{The
actual result is slightly weaker: see Section~\ref{sec:randerrors} for more details.}
Further, for \textit{most} codes of rate $1-H_q(\p)-\eps$, with high
probability, for a $\p$ fraction of random errors, the output list size is one. (This follows from
Shannon's famous result on the capacity of the $q$-ary symmetric channel:
for a proof, see e.g.~\cite{R-cocoon}.) It is also known that most
codes of rate $1-H_q(\p)-\eps$ have relative distance at least $\p$. 
Further, for $q\ge 2^{\Omega(1/\eps)}$, it is known that such a code cannot 
have distance more than $\p+\eps$:
this follows from the Singleton bound and the fact that for such an alphabet size, 
$1-H_q(\p) \ge 1-\p-\eps$ (cf.~\cite[Sec 2.2.2]{atri-thesis}).

\paragraph{Our Results.}
In our first main result, we show that the phenomenon above is universal, that is,
for \textit{every} $q$-ary code, with $q\ge 2^{\Omega(1/\eps)}$, the following property
holds: if the code has relative distance $\delta$, then for any $\p\le \delta -\eps$
fraction of random errors, with high probability, the Hamming ball of fractional 
radius $\p$ around the
received word will only have the transmitted codeword in it. We would like to point out three
related points. First, our result implies that if we relax the worst-case error model to a
random error model, then combinatorially one can always correct \textit{twice} the
number of errors. Second, one cannot hope to correct more than a $\delta$ fraction of random
errors: it is easy to see that, for instance, for Reed-Solomon codes, \textit{any} error pattern of
relative Hamming weight $\p>\delta$ will give rise to a list size greater than one. Finally, the
proof of our result follows from a fairly straightforward counting argument.

A natural follow-up question to our result is whether the lower bound of $2^{\Omega(1/\eps)}$
on $q$ can be relaxed. We show that if $q$ is $2^{o(1/\eps)}$, then the result above is \textit{not
true}.
This negative result follows from the following two observations/results. First, it is known that for 
\textit{any} code with rate $1-H_q(\p)+\eps$, the average list size, over all possible received
words, is exponential. Second, it is known that Algebraic-Geometric (AG) codes over alphabets
of size at least $49$ can have relative distance \textit{strictly} bigger than $1-H_q(\p)$ (cf.~\cite{handbook}). 
However, these two results do not immediately imply the negative result for the random error case. In
particular, what we need to show is that there is at least one codeword $\vc$ such that for
most error patterns $\ve$ of relative Hamming weight $\p$, the received word $\vc+\ve$ has at least one
codeword other than $\vc$ within a relative Hamming distance of $\p$ from it. To show that this can indeed be 
true for AG codes, we use a generalization of an ``Inverse Markov argument" from
Dumer et al.~\cite{DMS}.

\paragraph{A Cryptographic Application.} 
In addition to being a natural noise model to study, list decoding in the random error model
has applications in cryptography. In particular, Kiayias and Yung have proposed cryptosystems
based on the hardness of decoding Reed-Solomon codes~\cite{rs-crypto}. 
However, if for Reed-Solomon codes (of rate $R$), one can list decode $\p$ fraction of 
\textit{random} errors then the cryptosystem from~\cite{rs-crypto} can be broken for the
corresponding parameter settings. Since Guruswami-Sudan can solve this problem for $\p\le 1-\sqrt{R}$
for \textit{worst-case} errors~\cite{GS98}, Kiayias and Yung set the parameter $\p>1-\sqrt{R}$. Beyond the
$1-\sqrt{R}$ bound, to the best of our knowledge,
the only known algorithms to decode Reed-Solomon codes are the following trivial ones:
(i) Go through all possible $q^k$ codewords and output all the codewords with Hamming distance
of $\p$ from the received word; and (ii) Go through all possible $\binom{n}{\p n}$ error locations
and output the codeword, if any, that agrees in the $(1-\p)n$ ``non-error" locations.

It is interesting to note that each of the three algorithms mentioned above work in the
stronger model of worst-case errors. However, since we only care about decoding from random errors, one
might hope to design better algorithms that make use of the fact that the errors are random.
In this paper, we show that (essentially) the proof of our first main result
implies a related result that in turn implies a modest improvement in the running time of algorithms
to decode Reed-Solomon codes from $\p>1-\sqrt{R}$ fraction of random errors. The related result
states the following: for \text{any} code with relative distance $\delta$ (over a large enough
alphabet) with high probability, for a $\p$ fraction of random errors, Hamming balls of
fractional radius $\delta-\eps$ around the received word only have the transmitted codeword in them.\footnote{A similar result was shown for Reed-Solomon codes by McEliece~\cite{mceliece}.} Note that
unlike the statement of our result mentioned earlier, we are considering Hamming balls of radius
\textit{larger} than the fraction of errors. This allows us to improve the second trivial
algorithm in the paragraph above so that one needs to verify fewer ``error patterns."
This leads to an asymptotic improvement in the running time over both of the trivial
algorithms for certain setting of parameters, though the running time is still exponential
and thus, too expensive to
break the Kiayias-Yung cryptosystem.

\subsection{Erasures}

In the second part of the paper, we consider the erasure noise model, where the decoder knows the
locations of the errors. (However, the error locations are still chosen by
 the adversary.) Intuitively, this noise model
is weaker than the general worst-case noise model as the decoder knows for sure which locations are
uncorrupted. This intuition can also be formalized. E.g., it is known that for a $\p$ fraction of
worst-case errors, the list decoding capacity is $1-H_q(\p)$, whereas for a $\p$ fraction of
 erasures, the list decoding capacity is $1-\p$ (cf. \cite[Chapter 10]{G-thesis}).
Note that the capacity for erasures is \textit{independent} of the alphabet size. 
As another example, for a linear code, a combinatorial guarantee on list decodability from
erasures gives a polynomial time list decoding algorithm. By contrast, such a result is
not known for worst-case errors.

%


As is often the case, the capacity result is proven by random coding arguments.
A natural quest then is to design explicit linear codes that achieve the list decoding capacity for erasures, and is an important milestone in the program of designing explicit codes that achieve list decoding capacity for worst-case
errors. This goal is the primary motivation for our second main result.

\paragraph{Our Result and Related Work.} For large enough alphabets, explicit linear codes that achieve list
decoding capacity for erasures are not hard to find: e.g., Reed-Solomon codes achieve the capacity. For
smaller alphabets, the situation is much different. For binary codes, Guruswami
presented explicit linear codes that can handle $\p=1-\eps$ fraction of erasures with rate
$\Omega\left(\frac{\eps^2}{\log(1/\eps)}\right)$~\cite{Gur-erasures}. For alphabets of size $2^t$, $1-\eps$ fraction
of erasures can be list decoded with explicit linear codes of rate
 $\Omega\left(\frac{\eps^{1+1/t}}{t^2\log(1/\eps)}\right)$~\cite[Chapter 10]{G-thesis}. Thus,
especially for binary codes, an explicit code with capacity of $1-\p$ is still a lofty goal.
(In fact, breaking the $\eps^2$ rate barrier for \textit{polynomially} small $\eps$ would
imply explicit construction of certain bipartite Ramsey graphs, solving an open
question~\cite{Gur-erasures}.)

To gain a better understanding about codes that achieve list decoding capacity for erasures, a natural
question is to ask whether \textit{concatenated codes} can achieve the list decoding capacity for erasures.
Concatenated codes are the preeminent method to construct good list decodable codes over
small alphabets. In fact, the best explicit list decodable binary codes (for both erasures~\cite{Gur-erasures} and
worst-case errors~\cite{GR-multilevel}) are concatenated codes. Briefly, in code concatenation, an ``outer" code over a
large alphabet is first used to encode the message. Then ``inner" codes over the smaller alphabet
are used to encode each of the symbols in the outer codeword.  These inner codes typically have a 
much smaller block length than the outer code, which allows one to use brute-force type algorithms
to search for ``good" inner codes. Also note that the rate of the concatenated code is the
product of the rate of the outer and inner codes.

Given that concatenated codes have such a rigid structure, it seems plausible that such codes would
not be able to achieve list decoding capacity. For the worst-case error model, Guruswami and Rudra showed that there
do exist concatenated codes that achieve list decoding capacity~\cite{atri-venkat-soda-08}.
However, for erasures there is an additional potential complication that does not arise for the
worst-case error case. In particular, consider erasure patterns in which $\p$ fraction of the outer symbols are
completely erased. It is clear by this example that the outer code needs to have rate very close to $1-\p$. However,
note that to approach list decoding capacity for erasures, the concatenated code needs to have
rate $1-\p-\eps$. This means that the inner codes need to have rate very close to $1$.
By contrast, even though the result of~\cite{atri-venkat-soda-08} has some restrictions on the rate
of the inner codes, it is not nearly as stringent as the requirement above. 
(The restriction in~\cite{atri-venkat-soda-08} seems to be an artifact of the proof, whereas
for erasures, the restriction is unavoidable.)
Further, this restriction
on the inner rate is just by looking at a specific class of erasure patterns. It is reasonable to
wonder if when taking into account all possible erasure patterns, we can rule out the possibility
of concatenated codes achieving the list decoding capacity for erasures.

In our second main result, we show that concatenated codes \textit{can} achieve the list decoding
capacity for erasures. In fact, we show that choosing the outer code to be a Folded Reed-Solomon
code (\cite{GR-capacity}) and picking the inner codes to be random independent linear codes with rate $1$, will with high
probability, result in a linear code that achieves the list decoding capacity for erasures. We
show a similar result (but with better bounds on the list size) when the outer code is also chosen
to be a random linear code. Both of these ensembles were shown to achieve the list decoding
capacity for errors in~\cite{atri-venkat-soda-08}, although, as mentioned earlier, the result for errors
holds for a superset of concatenated codes (as the inner codes could have rates strictly less
than $1$). The proof of our result is similar to the proof structure in~\cite{atri-venkat-soda-08}.
Because we are dealing with the more benign erasure noise model, some of the calculations in
our proofs are much simpler than the corresponding ones in~\cite{atri-venkat-soda-08}.

\paragraph{Approximating NP Witnesses.} We conclude this section by pointing out that an application of
binary codes that are list decodable from erasures is to the problem of approximating NP-witnesses~\cite{GHLP,KS}.
For any NP-language $L$, we have a polynomial-time decidable relation $R_L(\cdot,\cdot)$ such that
$x\in L$ if and only if there exists a polynomially sized witness $w$ such that $R_L(x,w)$ accepts.
Thus, for an NP-complete language we do not expect to be able to compute the witness $w$ in
polynomial time
given $x$. A natural notion of approximation is the following: given an $\eps$ fraction of the bits
in a  
a correct witness $w$, can we verify if $x\in L$ in polynomial time? The results in
\cite{GHLP,KS} show that such an approximation is not possible unless P$=$NP.

To be more precise, G\'{a}l et al. (\cite{GHLP}) consider the following problem: given a SAT formula $\phi$ over
$n$ variables can we, in polynomial time, compute another SAT formula $\phi'$ over
$N=\mathrm{poly}(n)$ variables such that given $\eps N$ bits from a satisfying assignment to
$\phi'$, we can compute a satisfying assignment to the original formula $\phi$?

Kumar and Sivakumar's (\cite{KS}) reduction works for any NP-language
$L$. However, their reduction computes a polynomial-time computable relation $R'_L$ (with
witness size $N=\mathrm{poly}(n)$),
which is different from the original  predicate $R_L$ such that the knowledge of
$\eps N$ many bits of some satisfying witness for $R'_L$ can be used in polynomial time to compute
a satisfying witness for $R'_L$. Both of these results are proven by picking a
linear binary code $C$ that can be list decoded from a $1-\eps$ fraction of erasures and ``encoding"
$C(x)$ (where $x$ is the input) into the definition of $\phi'$ (in the case of~\cite{GHLP})
 or $R'_L$ (in the case of~\cite{KS}). The intuition behind these reductions
is that given sufficiently many bits
of a satisfying witness, we can obtain a list of potentially satisfying witnesses by running the list
decoding algorithm for $C$ to recover from the erasures. (The connection to list decoding was
implicit in~\cite{GHLP}-- it was made explicit in~\cite{KS}.)

Guruswami and Sudan (\cite{GS2000}) show that the reductions above can be made to work with $\eps=N^{-1/2+\g}$ for
the Kumar and Sivakumar problem and with $\eps=N^{-1/4+\g}$ for the G\'{a}l et al. problem (for any 
constant $\g>0$). An explicit linear code that meets the list decoding capacity for erasures will
improve the value of $\eps$ above to $N^{-1+\g}$ and $N^{-1/2+\g}$, respectively.

\paragraph{Organization of the Paper.} We begin with some preliminaries in Section~\ref{sec:prelims}.
We present our first main result on random codes in Section~\ref{sec:randerrors} and our
second main result on erasures in Section~\ref{sec:concatcodes}.

\section{Preliminaries}
\label{sec:prelims}
For an integer $m\ge 1$, we will use $[m]$ to denote the set $\{1,\dots,m\}$.

\paragraph{Basic Coding Definitions.}
A code $C$ of {\em dimension} $k$ and {\em block length} $n$ over an
alphabet $\Sigma$ is a subset of $\Sigma^n$ of size $|\Sigma|^k$. The
{\em rate} of such a code equals $k/n$.  Each $n$-tuple in $C$ is called a
codeword. 
Let $\F_q$ denote the field with $q$
elements. A code $C$ over
$\F_q$ is called a linear code if $C$ is a subspace of $\F_q^n$. In
this case the dimension of the code coincides with the dimension of
$C$ as a vector space over $\F_q$. By abuse of notation we can also
think of a linear code $C$ as a map from an element in $\F_q^k$ to its
corresponding codeword in $\F_q^n$, mapping a row vector $\vx \in \F_q^k$ to a vector
$\vx \vg \in \F_q^n$ via a $k\times n$ matrix $\vg$ over $\F_q$ which
is referred to as the generator matrix.

The Hamming distance between two vectors in $\vx,\vy\in\Sigma^n$, denoted by $\Delta(\vx,\vy)$,
 is the number
of places they differ in. The (minimum) distance of a code $C$ is the
minimum Hamming distance between any two distinct codewords
from $C$. The relative distance is the ratio of the distance to the
block length.

We will need the following notions of the weight of a vector.  Given a
vector $\myvec{v}\in\{0,1,\dots,q-1\}^{n}$, its Hamming weight, which is the
number of non-zero entries in the vector, is denoted by
$\wt(\myvec{v})$.  Given a vector $\myvec{y}=(y_1,\dots,y_n)\in
\{0,\dots,q-1\}^n$ and a subset $S\subseteq [n]$, $\vy_S$ will denote
the subvector $(y_i)_{i\in S}$, and 
$\wt_S(\myvec{y})$ will denote the Hamming weight of $\vy_S$.


\paragraph{Code Concatenation.}
Concatenated codes are constructed from two different types of codes
that are defined over alphabets of different sizes.  If we are
interested in a concatenated code over $\F_q$,
then the \textit{outer code}
$\cout$ is defined over $\F_Q$, where $Q=q^k$ for some positive
integer $k$, and has block length $N$. The second type of codes, called
the \textit{inner codes}, and which are denoted by $\cin^1,\dots,\cin^N$,
are defined over $\F_q$ and are each of dimension $k$ (note that the
message space of $\cin^i$ for all $i$ and the alphabet of $\cout$ have
the same size).  The concatenated code, denoted by
$C=\cout\circ(\cin^1,\dots,\cin^N)$, is defined as follows: Let the
rate of $\cout$ be $R$ and let the block lengths of 
$\cin^i$ be $n$ (for $1\le i\le N$). Define
$K=RN$ and $r=k/n$. The input to $C$ is a vector
$\vm=\angles{m_1,\dots,m_K}\in (\F_q^k)^K$. Let
$\cout(\vm)=\angles{x_1,\dots,x_N}$.  The codeword in $C$
corresponding to $\vm$ is defined as follows
\[C(\vm)=\angles{\cin^1(x_1),\cin^2(x_2),\dots,\cin^N(x_N)}.\]

The outer code $\cout$ in this paper will either be a random linear code over $\F_Q$
or the folded Reed-Solomon code from~\cite{GR-capacity}.  In the case
when $\cout$ is random linear, we will pick $\cout$ by selecting $K=RN$
vectors uniformly at random from $\F_Q^N$ to form the rows of the
generator matrix. For every position $1\le i\le N$, we will choose an
inner code $\cin^i$ to be a random linear code over $\F_q$ of block
length $n$ and rate $r=k/n$. In particular, we will work with the
corresponding generator matrices $\myvec{G}_i$, where every
$\myvec{G}_i$ is a random $k\times n$ matrix over $\F_q$. All the
generator matrices $\myvec{G}_i$ (as well as the generator matrix for
$\cout$, when we choose a random $\cout$) are chosen independently.
This fact will be used crucially in our proofs.

\paragraph{List Decoding.} We define some terms related to list decoding.
\begin{definition}[List decodable code for errors]
For $0 < \rho <1$ and an integer $L \ge 1$, a code $C \subseteq
\Sigma^n$ is said to be $(\rho,L)$-list decodable if for every
$\vy\in\Sigma^n$, the number of codewords in $C$ that are within Hamming
distance $\rho n$ from $y$ is at most $L$.
\end{definition}

Given a vector $\vc=(c_1,\dots,c_n)\in\Sigma^n$ and an erased received word $\vy=(y_1,\dots,y_n)\in(\Sigma\cup\{?\})^n$,\footnote{$?$ denotes an erasure.} we will use
$\vc\simeq \vy$ to denote the fact that for every $i\in [n]$ such that $y_i\neq ?$, $c_i=y_i$. With this
definition, we are ready to define the notion of list decodability for erasures. Further, for an erased received
word, we will use $\wt(\vy)$ to denote the number of erased positions.

\begin{definition}[List decodable code for erasures]
For $0 < \rho <1$ and an integer $L \ge 1$, a code $C \subseteq
\Sigma^n$ is said to be $(\rho,L)_{led}$-list decodable if for every
$\vy\in(\Sigma\cup\{?\})^n$ with $\wt(\vy)\le \p n$, the number of codewords $\vc\in C$
such that $\vc\simeq \vy$ is at most $L$.
\end{definition}

\paragraph{Reed-Solomon and Related Codes.}
The classical family of Reed-Solomon (RS) codes over a field $\F$ are
defined to be the evaluations of low-degree polynomials at a sequence
of distinct points of $\F$. Folded Reed-Solomon codes are obtained by
viewing the RS code as a code over a larger alphabet $\F^s$ by
bundling together $s$ consecutive symbols for some folding parameter
$s$.  
We will not need any specifics of folded RS codes (in fact, even their
definition) beyond certain properties that we recall in Section~\ref{sec:concatcodes}.

\section{Random Errors}
\label{sec:randerrors}

In this section we consider the random noise model mentioned in the introduction:
the error locations are adversarial but the errors themselves are random.
Our main result is the following.

\begin{theorem}
\label{thm:dist-list-one}
Let $0<\eps,\delta<1$ be reals and let $q$ and $n\ge \Omega(1/\eps)$ be positive integers.
Let $\Sigma=\{0,1,\dots,q-1\}$.\footnote{We will assume that $\Sigma$ is equipped with a
monoid structure,
i.e. for any $a,b\in\Sigma$, $a+b\in \Sigma$ and $0$ is the identity element.} Let $0<\p\le\delta-\eps$ be a real.
Let $C$ be a code over $\Sigma$ of block length $n$ and
relative distance $\delta$.
Let
$S\subseteq [n]$ with $|S|= (1-\p)n$. Then the following hold:
\begin{itemize}
\item[(a)]
If $q\ge 2^{\Omega(1/\eps)}$, then
for every codeword $\vc$ and all but a $q^{-\Omega(\eps n)}$ fraction of
error patterns $\ve\in \Sigma^n$ with $\wt(\ve)=\p n$ and $\wt_S(\ve)=0$, the
only codeword within the Hamming ball of radius $\p n$ around the
received word $\vc+\ve$ is $\vc$.
\item[(b)] Let $\g>0$.
If $q> \max\left(n,\left(\frac{e}{1-\delta+\eps}\right)^{\left\lceil \frac{1}{\g}\right\rceil}\right)$, then
for every codeword $\vc$ and all but a $(q-1)^{-((1-\g)\eps/2-(1-\delta)\g)n}$ fraction of
error patterns $\ve\in \Sigma^n$ with $\wt(\ve)=\p n$ and $\wt_S(\ve)=0$, the
only codeword within the Hamming ball of radius $(\delta-\eps) n$ around the
received word $\vc+\ve$ is $\vc$.
\end{itemize}
\end{theorem}

A weaker version of Theorem~\ref{thm:dist-list-one} was previously known for RS codes~\cite{mceliece}.
(Though the bounds for part (b) are better in~\cite{mceliece}.) In particular, McEliece showed Theorem~\ref{thm:dist-list-one}
for RS codes but over \textit{all} error patterns of Hamming weight $\p n$. In other words, Theorem~\ref{thm:dist-list-one}
implies the result in~\cite{mceliece} if we average our result over all subsets $S\subseteq [n]$ with
$|S|=\p n$.

Part (a) of Theorem~\ref{thm:dist-list-one} implies that for $e\le (\delta-\eps)n$ random errors,
with high probability, the Hamming ball of radius $e$ has one codeword in it. Note that this
is \textit{twice} the number of errors for which an analogous result can be shown
for worst-case errors.
Part (b) of Theorem~\ref{thm:dist-list-one} implies the following property
of Reed-Solomon codes (where we pick $\eps=4R$ and $\g=1/2$).



\begin{corollary}
\label{cor:rs-avg-list-one}
Let $k\le n< q$ be integers such that $q>\left(\frac{n}{k}\right)^2$. Then
the following property holds for Reed-Solomon codes of dimension $k$ and
block length $n$ over $\F_q$. For at least $1-q^{-\Omega(k)}$ fraction of
error patterns $\ve$ of Hamming weight {\em at most} $n-4k$ and any codeword
$\vc$, the only codeword that agrees in at least $4k$ positions with
$\vc+\ve$ is $\vc$.
\end{corollary}

We would like to point out that in Corollary~\ref{cor:rs-avg-list-one}, the radius of the Hamming
ball can be \textit{larger} than the number of errors. This can be used to slightly improve 
upon the best known algorithms to decode RS codes from random errors beyond the Johnson
bound for super-polynomially large $q$. See Section~\ref{sec:random-algo} for 
more details.

A natural question is whether the lower bound of $q\ge 2^{\Omega(1/\eps)}$ in part (a) of Theorem
\ref{thm:dist-list-one} can be improved. In Section~\ref{sec:alphabet-lb} we show that this is
not possible.

\paragraph{Proof of Theorem~\ref{thm:dist-list-one}.}
Let $\vc\in C$ be the transmitted codeword. For an $\a\ge 1-\delta+\eps$,
we call an error pattern
$\ve$ (with $\wt(\ve)=\p n$ and $\wt_S(\ve)=0$) $\a$-\textit{bad} if there exits 
a codeword
$\vc'\neq\vc\in C$ such that $\Delta(\vc+\ve,\vc')= (1-\a)n$ (and every
other codeword has a larger Hamming distance from $\vc+\ve$). We will 
show that the number of $\a$-bad error patterns (over all
$\a\ge 1-\delta+\eps$) is an exponentially
small fraction of error patterns $\ve$ with $\wt(\ve)=\p n$ and $\wt_S(\ve)=0$,
which will prove the theorem.

Fix $\a\ge 1-\delta+\eps$. Associate every $\a$-bad error pattern $\ve$ with the 
lexicographically first codeword $\vc'\neq\vc\in C$ such that $\Delta(\vc+\ve,\vc')=(1-\a)n$. 
Let $A\subseteq [n]$ be the set of positions where
$\vc'$ and $\vc+\ve$ agree. Further, define $S_0=S\cap A$, $S_1=A\cap ([n]\setminus S)$
and $\b=|S_0|/n$.
Thus, for every $\a$-bad error pattern $\ve$, we can associate such a pair of
subsets $(S_0,S_1)\subseteq S\times ([n]\setminus S)$. Hence, to count the number of $\a$-bad
error patterns it suffices to count for each possible pair $(S_0,S_1)$,
with $|S_0|=\b n$ and $|S_1|=(\a -\b)n$ for some $\a-\p\le \b\le \a$, the
number of $\a$-bad patterns that can be associated with it. (The lower and upper bounds on
$\b$ follow from the fact that $S_1\subseteq [n]\setminus S$ and $S_0\subseteq A$,
respectively.)

Fix sets $S_0\subseteq S$ and $S_1\subseteq [n]\setminus S$
with $|S_0|=\b n$ and $|S_1|=(\a-\b)n$ for some $\a-\p\le \b\le \a$. 
To upper bound the number of $\a$-bad 
error patterns that are associated with $(S_0,S_1)$, first note that such
error patterns take all the $(q-1)^{(\p-\a+\b)n}$ possible
values at the positions in $[n]\setminus(S\cup S_1)$. Fix a vector
$\vx$ of length $n-|S|-|S_1|$ and consider all the $\alpha$-bad
error patterns $\ve$
such that $\ve_{[n]\setminus (S\cup S_1)}=\vx$.
Recall that each error pattern is associated with a codeword $\vc'\neq \vc$
such that $\vc'$ and $\vc+\ve$ agree exactly in the positions $S_0\cup S_1$.
Further, such a codeword $\vc'$ is associated with exactly one $\a$-bad 
error pattern $\ve$, where $\ve_{[n]\setminus (S\cup S_1)}=\vx$. (This is
because fixing $\vc'$ fixes $\ve_{S_1}$ and $\ve_S$ is already fixed by
the definition of $S$.) Thus, to upper bound the number of $\a$-bad 
error patterns associated with $(S_0,S_1)$, where $\ve_{[n]\setminus (S\cup S_1)}=\vx$ (call this number $N_{\a,S_0,S_1,\vx}$),
we will upper bound the number of such codewords $\vc'$. Note that as $C$ has
relative distance $\delta n$, once any  $(1-\delta)n+1$ positions are fixed,
there is at most one codeword that agrees with the fixed positions (if there
is no such codeword then the corresponding ``error pattern" does not exist).
Thus, there is at most one possible
$\vc'$  once we fix (say) the ``first" $(1-\delta)n+1-|S_0|$
values of $\ve_{S_1}$ (recall that $\vc'_{S_0}=\vc_{S_0})$. This implies that
\[N_{\a,S_0,S_1,\vx}\le (q-1)^{(1-\delta-\b)n+1}.\]
Let $M_{\a}$ be the  number of choices for $(S_0,S_1)$, which is just the number of
choices for $A$. As
the number of choices for $\vx$ is $(q-1)^{(\p-\a+\b)n}$,  the
number of $\a$-bad error patterns is at most
\begin{equation}
\label{eq:num-bad-pattern}
M_{\a}\cdot (q-1)^{(\p-\a+\b)n}\cdot (q-1)^{(1-\delta-\b)n+1}=M_{\a}\cdot(q-1)^{(1-\delta-\a)n+1}\cdot (q-1)^{\p n}.
\end{equation}

\paragraph{Proof of part(a).} Note that the number of $\a$-bad patterns for
any $\a\ge 1-\delta+\eps$ is upper bounded by
\[M_{\a}\cdot(q-1)^{-\eps n+1}\cdot (q-1)^{\p n}.\]
We trivially upper bound $M_{\a}$ by $2^n$.
Recalling that there are $(q-1)^{\p n}$ error patterns $\ve$ with $\wt(\ve)=\p n$
and $\wt_S(\ve)=0$ and that $\a$ can take at most $n$ values, the fraction
of $\a$-bad patterns (over all $\a\ge 1-\p\ge 1-\delta+\eps$) is at most
\[n2^{n}(q-1)^{-\eps n+1}\le (q-1)^{\left(-\eps +\frac{2}{\log(q-1)}+\frac{1}{n}\right)n}\le (q-1)^{-\eps n/3}\le q^{-\eps n/6},\]
where the first inequality follows from the fact that $n\le 2^n$, the
second inequality is true for $n\ge 3/\eps$ and $q\ge 2^{6/\eps}$ and the last 
inequality follows from the inequality
$(q-1)\ge \sqrt{q}$ (which in turn is true for $q\ge 3$).

\paragraph{Proof of part (b).} Note that $M_{\a}=\binom{n}{\a n}\le (e/\a)^{\a n}$. Thus, the number of $\a$-bad error patterns is upper bounded by
\[ (q-1)^{\left(1-\delta-\a+\a\cdot\frac{\log(e/\a)}{\log(q-1)}\right)n+1}\cdot (q-1)^{\p n}
\le (q-1)^{(1-\delta-\a(1-\g))n+1}\cdot (q-1)^{\p n}\le 
(q-1)^{(-(1-\g)\eps+\g(1-\delta))n+1}\cdot (q-1)^{\p n},
\]
where the inequalities follow from the facts that $q> \left(\frac{e}{1-\delta+\eps}\right)^{1/\g}$ and $\a\ge 1-\delta+\eps$.
Recalling that there are $(q-1)^{\p n}$ error patterns $\ve$ with $\wt(\ve)=\p n$
and $\wt_S(\ve)=0$ and that $\a$ can take at most $n$ values, the fraction
of $\a$-bad patterns (over all $\a\ge 1-\delta+\eps$) is at most
\[n(q-1)^{(-(1-\g)\eps+(1-\delta)\g)n+1}\le (q-1)^{\left(-(1-\g)\eps+\g(1-\delta)+\frac{2}{n}\right)n}\le 
(q-1)^{\left(-\frac{(1-\g)\eps}{2}+\g(1-\delta)\right)n},\]
where the first inequality follows from the fact that $q>n$ and the
second inequality is true for $n\ge 4/((1-\g)\eps)$.
\QED

\subsection{An Implication of Corollary~\ref{cor:rs-avg-list-one}}
\label{sec:random-algo}

To the best of our knowledge,
for $e>n-\sqrt{kn}$, the only known algorithms to decode Reed-Solomon (RS) codes
from $e$ random errors are the trivial ones:
(i) Go through all possible codewords and output the closest codeword--
this takes $2^{O(k\log{q})}\cdot n$ time and (ii) Go through all
possible $\binom{n}{e}$ error locations and check that the received word
outside the purported error locations is indeed a RS codeword-- this
takes $2^{O((n-e)\log(n/(n-e)))}\cdot O(n^2)$ time.

If $e\le n-4k$, then by Corollary~\ref{cor:rs-avg-list-one}, we can go through all
the $\binom{n}{4k}$ choices of subsets of size $4k$ and check if the received word
projected down to the subset lies in the corresponding projected down RS code.
This algorithm takes $2^{O(k\log(n/k))}\cdot O(n^2)$ time, which is better than the
trivial algorithm (ii) mentioned above for $e$ in $n-\omega(k)$.
Further, this algorithm is better than the trivial algorithm (i) when $q$ is
super-polynomially large in $n$.

\subsection{On the Alphabet Size in Theorem~\ref{thm:dist-list-one}}
\label{sec:alphabet-lb}

It is well-known that \textit{any} code that is $(\p,L)$-list decodable that also has rate
at least $1-H_q(\p)+\eps$ needs to satisfy $L=q^{\Omega(\eps n)}$ (cf.~\cite{G-thesis}).
A natural way to try to show that part (a) of Theorem~\ref{thm:dist-list-one} is false for
$q\le 2^{o(1/\eps)}$ is to look at codes whose relative distance is strictly larger than
$1-H_q(\p)$. Algebraic-geometric (AG) codes are a natural candidate since they can 
beat the Gilbert-Varshamov bound for an alphabet size of at least $49$ (cf. \cite{handbook}). The only catch is that the lower bound on
$L$ follows from an average case argument and we need to show that over \textit{most} error patterns,
the list size is more than one. For this we need an ``Inverse Markov argument," like one in
~\cite{DMS}. 

(The argument above was suggested to us by Venkat Guruswami.)

We begin with the more general statement of the ``Inverse Markov argument" from~\cite{DMS}. (We thank
Madhu Sudan for the statement and its proof.)

\begin{lemma}
\label{lem:inverse-markov}
Let $G=(L,R,E)$ be a bipartite graph with $|L|=n_L$ and $|R|=n_R$. Let the average left degree of $G$ be denoted by $\bar{d_L}$.
Note that the average right degree is $\bar{d_R}=\frac{n_L\cdot d_L}{n_R}$.
Then the following statements are true:
\begin{itemize}
\item[(i)] If we pick an edge $e=(u,v)$ uniformly at random from $E$, then the probability that\footnote{For any vertex $v$, we denote its degree by $d(v)$.}
$d(v)\le \eps \bar{d_R}$ is at most $\eps$.
\item[(ii)] If $G$ is $d$-left regular then consider the following process: Uniformly at random
pick a vertex $u\in L$. Then uniformly at random pick a vertex $v\in R$ in $u$'s neighborhood. Then the
probability that $d(v)\le \eps \frac{ d n_L}{n_R}$ is at most $\eps$.
\end{itemize}
\end{lemma}
\begin{proof}
We first note that (ii) follows from (i) as the random process in (ii) ends up picking edges uniformly
at random from $E$.

To conclude, we prove part (i). Consider the set $R'\subseteq R$ such that $v\in R'$ satisfies
$d(v)\le \eps \bar{d_R}$. Note that that the maximum number of edges that have an end-point in
$R'$ is at most $\eps \bar{d_R}\cdot n_R = \eps |E|$. Thus, the probability that a uniformly random
edge in $E$ has an end point in $R'$ is upper bounded by $\eps |E|/|E| =\eps$, as desired.
\end{proof}

The following is an easy consequence
of Lemma~\ref{lem:inverse-markov} and the standard probabilistic method used to prove the lower bound
for list decoding capacity. 

\begin{lemma}
\label{lem:av-bad-ball}
Let $q\ge 2$ and $0\le \p < 1-1/q$. Then the following holds for large enough $n$.
Let $C\subseteq \{0,\dots,q-1\}^n$ be a code with rate $1-H_q(\p)+\g$. Then there exists a codeword $\vc\in C$
such that for at least a $1-q^{-\Omega(\g n)}$ fraction of error patterns $\ve$ of Hamming weight at most
$\p n$,
it is true that the Hamming ball of radius $\p n$ around $\vc +\ve$ has at least two codewords from
$C$ in it.
\end{lemma}
\begin{proof}
Define the bipartite graph $G_{C,\p}=(C,\{0,\dots,q-1\}^n,E)$
as follows. For every $\vc\in C$, add $(\vc,\vy)\in E$ such that $\Delta(\vc,\vy)\le \p n$. Note that
$G_{C,\p}$ is a $\vol{q}(\p n)$-left regular bipartite graph, where 
$\vol{q}(r)$ is the volume of the $q$-ary Hamming ball with radius $r$. Note that the graph has an
average right degree of
\[\bar{d_R}= \frac{\vol{q}(\p n) \cdot q^{(1-H_q(\p)+\g)n}}{q^n} \ge q^{\g n-o(n)},\]
where in the above we have used the following well known inequality (cf.~\cite{MS}):
\[\vol{q}(\p n)\ge q^{H_q(\p )n-o(n)}.\]
Thus, by part (b) of Lemma~\ref{lem:inverse-markov} (with $\eps=(\bar{d}_R)^{-1}\le q^{-\g n+o(n)}$), we have
\[\pr_{\vc\in C}\pr_{\substack{\ve\in \{0,\dots,q-1\}^n\\\wt(\ve)\le \p n}}\left[ \vc+\ve\text{ has at most one  codeword within Hamming distance } \p n\right] \le q^{-\g n +o(n)}.\]
Thus, there must exist at least one codeword $\vc\in C$ with the required property.
\end{proof}

Thus, given Lemma~\ref{lem:av-bad-ball},  we can prove that part (a) of Theorem~\ref{thm:dist-list-one}
is not true for a certain value of $q$ if there exists a code $C\subseteq \{0,\dots,q-1\}^n$ with relative
distance $\delta$ such that it has rate at least $1-H_q(\delta -\eps)+\g$ for some $\g >0$. Now it
is known that for fixed $\alpha >0$, $H_q(\alpha) \ge \alpha + \Omega\left(\frac{1}{\log{q}}\right)$
(cf.~\cite[Lecture 7]{sudan-notes}).
Thus, we would be done if we could find a code with relative distance $\delta$ and rate at least
\[1-\delta+\eps +\g -O(1/\log{q}).\]
For $q\le 2^{o(1/\eps)}$, the bound above for small enough $\eps$ is upper bounded by 
$1-\delta -\eps-\frac{1}{\sqrt{q}-1}$
(assuming that $\g=\Theta(\eps)$).
It is known that AG codes over alphabets of size $\ge 49$ with relative distance $\delta$ exist that achieve a rate of $1-\delta-\frac{1}{\sqrt{q}-1}$. Thus, for $49\le q\le 2^{o(1/\eps)}$, AG codes over alphabets of size
$q$ are the required codes.

\section{Concatenated Codes}
\label{sec:concatcodes}

This section first shows that with folded Reed-Solomon codes and
independently chosen small random linear inner codes, the resulting
concatenated code can achieve erasure capacity in a list decoding setting.
A similar result holds when the outer code is a random linear code,
and this result is presented second. 

\subsection{Folded Reed-Solomon Outer Code}
\begin{theorem}
\label{thm:frs-concat}
Let $q$ be a prime power and let $0<R \le 1$ be an arbitrary rational number.
Let $n,K,N\ge 1$ be large enough integers such that $K=RN$.
%
Let $\cout$ be a folded Reed-Solomon code over $\F_{q^{n}}$ of block
length $N$ and rate $R$. Let $\cin^1,\dots,\cin^N$ be random linear
codes over $\F_q$, where $\cin^i$ is generated by a random $n\times n$
matrix $\vg_i$ over $\F_q$ and the random choices for
$\vg_1,\dots,\vg_N$ are all independent.\footnote{We stress that we do
  {\em not} require that the $\vg_i$'s have rank $n$.}%
Then the concatenated code $C^*=\cout\circ(\cin^1,\dots,\cin^N)$ is a
$\left(1-R-\eps,\left(\frac{N}{\eps^2}\right)^{O\left(\eps^{-2}\log(1/R)\right)}\right)_{led}$-list
decodable code with probability at least $1-q^{-\Omega(nN)}$ over the
choices of $\vg_1,\dots,\vg_N$. Further, $C^*$ has rate $R$ w.h.p.
\end{theorem}

To set up the proof of the theorem above,
we begin by collecting certain definitions and results from~\cite{atri-venkat-soda-08}.
The following notion of independence will be crucial.

\begin{definition}[Independent tuples]
Let $C$ be a code of block length $N$ and rate $R$  defined over $\F_{q^k}$. Let
$J\ge 1$ and $0\le d_1,\dots,d_J\le N$ be integers. Let
$\vd=\angles{d_1,\dots,d_J}$.
An ordered tuple of codewords $(\vc^1,\dots,\vc^J)$, $\vc^j\in C$ is said
to be $(\vd,\F_q)$-independent if the following holds.
$d_1=\wt(\vc^1)$ and for every $1<j\le J$, $d_j$ is the number
of positions $i$ such that $c_i^j$ is $\F_q$-independent of the vectors 
$\{c^1_i,\dots,c^{j-1}_i\}$, where $\vc^{\ell}=(c^{\ell}_1,\dots,c^{\ell}_N)$.
\end{definition}

\noindent
Note that for any tuple of codewords $(\vc^1,\dots,\vc^J)$ there exists a
unique $\vd$ such that it is $(\vd,\F_q)$-independent. The next two results
will be crucial in the  proof of our second main result.

\begin{lemma}[\cite{atri-venkat-soda-08}]
\label{lem:frs-indp}
Let $\eps>0$ and
let $C$ be a folded Reed-Solomon code of block length $N$ and rate $0<R<1$ that is defined over
$\F_Q$, where $Q=q^k$. 
For any $L$-tuple
of codewords from $C$, where $L\ge
J\cdot (N/{\eps^2})^{O\left(\eps^{-1} J\log(q/R)\right)}$,
there exists a sub-tuple of 
$J$ codewords such that 
the
$J$-tuple is $(\vd,\F_q)$-independent, where $\vd=\angles{d_1,\dots,d_J}$
with  $d_j\ge (1-R-\eps)N$, for every $1\le j\le J$.
\end{lemma}

\begin{lemma}[\cite{atri-venkat-soda-08}]
\label{lem:frs-indp-count}
Let $C$ be a folded Reed-Solomon code of block length $N$ and rate $0<R<1$ that is defined over
$\F_Q$, where $Q=q^k$. 
Let $J\ge 1$ and $0\le d_1,\dots,d_J\le N$ be integers and define $\vd=
\angles{d_1,\dots,d_J}$. Then the number of $(\vd,\F_q)$-independent
tuples in $C$ is at most 
\[ q^{NJ(J+1)}\prod_{j=1}^J Q^{\max(d_j-N(1-R)+1,0)} \ . \]
\end{lemma}

Given the outer code $\cout$ and the inner codes $\cin^i$, recall that
for every codeword $\myvec{u}=(\vu_1,\dots,\vu_N)\in\cout$, the
codeword  $\vu\vg\stackrel{def}{=}(\vu_1\vg_1,\vu_2\vg_2,\dots,\vu_N\vg_N)$ is in $C^{*}=\cout\circ(\cin^1,\dots,\cin^N)$, 
where the operations are over $\F_q$.

We now begin with the proof.
The fact that $C^*$
has rate $R$ w.h.p. follows the argument used in~\cite{atri-venkat-soda-08} and is omitted. 

Define $Q=q^k$. 
Let $L$ be the worst-case list size that we are
aiming for (we will fix its value at the end). By Lemma~\ref{lem:frs-indp},
any $L+1$-tuple of $\cout$ codewords $(\vu^0,\dots,\vu^L)\in(\cout)^{L+1}$ 
contains at least 
$J= \left\lfloor (L+1)/(N/{\gamma^2})^{O\left(\gamma^{-1} J\log(q/R)\right)}\right\rfloor$ codewords that
form a $(\vd,\F_q)$-independent tuple, for some $\vd=\angles{d_1,\dots,d_J}$,
with $d_j\ge (1-R-\gamma)N$ for all $1 \le j \le J$ (we will specify $\g$, $0<\gamma<1-R$, later). 
Thus, to prove the theorem it suffices to show that with high
probability, there is no received word $\vy\in (\F_q\cup\{?\})^{nN}$ with $\wt(\vy)\le (1-R-\eps)nN$
and $J$-tuple of codewords $(\vu^1\vg,\dots,\vu^J\vg)$, where
$(\vu^1,\dots,\vu^J)$ is a $J$-tuple of folded Reed-Solomon codewords
that is $(\vd,\F_q)$-independent, such that $\vu^i\vg\simeq \vy$ for every $1\le i\le J$.
For the rest of the proof, we will call a $J$-tuple of $\cout$ codewords
$(\vu^1,\dots,\vu^J)$ a \textit{good} tuple if it is $(\vd,\F_q)$-independent
for some $\vd=\angles{d_1,\dots,d_J}$, where $d_j\ge (1-R-\gamma)N$ for every
$1\le j\le J$.

Define $\p=1-R-\eps$. Note that by the union bound, we need to show that
\begin{equation}
\label{eq:main-frs-concat}
\sum_{\substack{\vy\in(\F_q\cup\{?\})^{nN}\\\wt(\vy)\le \p nN}} P_{\vy}
\le q^{-\Omega(nN)},
\end{equation}
where
\[P_{\vy}=
\sum_{\textrm{good }(\vu^1,\dots,\vu^J)\in(\cout)^J}
\pr\left[ \bigwedge_{i=1}^J \vu^i\vg\simeq \vy\right].
\]

For now fix a good tuple $(\vu^1,\dots,\vu^J)$ that is $(\vd=\langle d_1,\dots,d_J\rangle,\F_q)$-independent.
Define sets $S_i\subseteq [N]$ ($|S_i|=d_i$) to be the positions that are ``witnesses" to the fact that $(\vu^1,\dots,\vu^J)$ is
$(\vd,\F_q)$-independent.

Then the probability that a particular codeword matches the unerased positions of the received word is:

\begin{equation}\label{trickery}
\pr[\vu^i\vg\simeq \vy] \leq \pr[(\vu^i\vg)_{S_i} \simeq \vy_{S_i}].
\end{equation}

Further, the latter probability in inequality \eqref{trickery} is
independent of the probability for any  $j\neq i$.

To see this, let $E_i$ be the event that $(\vu^i\vg)_{S_i} \simeq \vy_{S_i}$.

Then note that: $$\pr\left[\bigwedge_{i=1}^J E_i\right]=\pr\left[\bigwedge_{i=2}^J E_i\mid E_1\right]\cdot \pr[E_1].$$

As $(\vu^1,\dots,\vu^J)$ is a good tuple, this is simply:

$$=\pr\left[\bigwedge_{i=2}^J E_i\right]\cdot \pr[E_1].$$

Using induction, we get that the probability that all messages in the list match is just the product of the individual probabilities. Thus, we have:

$$\pr\left[\bigwedge_{i=1}^J \vu^i\vg \simeq \vy\right] \leq 
\pr\left[\bigwedge_{i=1}^J (\vu^i\vg)_{S_i} \simeq\vy_{S_i}\right] = \prod_{i=1}^J \pr[(\vu^i\vg)_{S_i} \simeq \vy_{S_i}].$$

If we let $u_i$ be the number of unerased $q$-ary symbols in $\vy_{S_i}$, then since all the $\vg_i$ are independent
random matrices:

$$\pr[(\vu^i\vg)_{S_i} \simeq \vy_{S_i}] = q^{-u_i} \leq q^{-d_i n + \rho n N}.$$

Note that the reason that $(-d_i n + \rho n N) \geq -u_i$ is because in
the worst case, all erasures occur in $S_i.$

We take a union bound over the number of different ways that the $d_i$ can
occur:

\begin{equation}\label{mother}
P_{\vy}\leq\sum_{(1-R-\g)N\le d_1,d_2,\cdots,d_J\le N} \left(q^{NJ(J+1)} \prod_{i=1}^J Q^{~\max{(0,d_i-N(1-R))}}\right)\prod_{i=1}^J q^{-d_i n + \rho n N}.
\end{equation}

The bound in parenthesis in inequality \eqref{mother} comes from Lemma~\ref{lem:frs-indp-count}.

Now since $$\max{(0,d_i-(1-R)N)} \leq d_i - (1-R-\gamma) N,$$

we can rewrite this, collapsing the two products into one, as:

\begin{equation}
\label{eq:small-r}
P_{\vy}=\sum_{(1-R-\g)N\le d_1,d_2,\cdots,d_J\le N} \left(q^{NJ(J+1)} \prod_{i=1}^J q^{n(d_i -(1-R-\gamma) N) - d_i n + \rho n N}\right).
\end{equation}

But since:

$$n d_i - n d_i = 0,$$

we can rewrite this again, replacing the sum with an upper bound, as

$$P_{\vy} \leq N^J q^{NJ(J+1)} q^{JnN (\rho -1 + R + \gamma)}. $$

Note that: 

$$q^{NJ(J+1)} = q^{N J n \left(\frac{J+1}{n}\right)}.$$

So for $n\ge (J+1)/\g$:

$$q^{NJ(J+1)} \leq q^{JnN\gamma}.$$

Note also that the total number of possible received words can be bounded as follows:

\begin{equation}\label{brick}
{{nN}\choose{\rho nN}}\cdot q^{(1-\p)nN} \leq q^{2nN},
\end{equation}

where the first term in the product on the left-hand side of inequality
\eqref{brick} is the number of ways to choose erasure locations, and the
second term is the number of ways to choose symbols in the unerased positions.

Also,

$$N^J\le q^{J\log N}\leq q^{JnN\gamma}$$

for large enough $N$.

After applying these bounds, we get that:

\begin{equation}\label{fudge}
\pr[C^*~\text{is not}~(\rho,L)_{led}] \leq q^{2nN} q^{J n N (\rho -1 +R + 3\gamma)}.
\end{equation}

Recall that we have $R=1-\rho-\epsilon$ and can choose $J$ and $\gamma$ freely.

Setting

$$J\geq 1/\gamma$$

will make

$$q^{nN} \leq q^{J n N \gamma},$$

and in particular,

$$q^{2nN} \leq q^{J n N (2\gamma)}.$$

If we pick $\gamma=\epsilon/10$, then our final
error probability in inequality \eqref{fudge} will be:

$$\pr[C^*~\text{is not}~(\rho,L)_{led}] \leq q^{-\frac{\epsilon}{2} J n N},$$

establishing the desired error bound.

\begin{remark}
It is easy to see that the rate of the inner codes have to be very close to $1$.
To see this consider the erasure pattern where $\rho$ fraction of the outer codeword
symbols are completely erased. To recover from such a situation, we need $R$ to be
close to $1-\rho$. One could re-visit the proof above for general $r$ and
try to figure out how far away from $1$ $r$ can be. If we had $r<1$ then in
(\ref{eq:small-r}), the exponent within the product should read $rn(d_i-(1-R-\gamma))-d_in+\rho nN$. We ultimately need $R^*=Rr=1-\rho-\epsilon$. Using this and some
manipulations, the exponent becomes $(1-r)(1-d_i/N)-\epsilon +r\gamma$. The only thing that we
can guarantee about $d_i$ is that $d_i\ge (1-R-\gamma)N$. If we desire the ultimate
error probability to be $q^{-\Omega(\epsilon nNJ)}$, then the proof goes
through only if $rR\ge R-O(\epsilon)$.
\end{remark}

\subsection{Random Linear Outer Code}

\begin{theorem}
\label{thm:random}
Let $q$ be a prime power and let $0<R \le 1$ be an arbitrary rational.
Let $n,K,N\ge 1$ be large enough integers such that $K=RN$.
Let $\cout$ be a random linear code over $\F_{q^{n}}$
that is generated by a random $K\times N$ matrix over $\F_{q^n}$.
Let
$\cin^1,\dots,\cin^N$ be random linear codes over $\F_q$, where
$\cin^i$ is generated by a random $n\times n$ matrix $\vg_i$
and the random choices for $\cout,\vg_1,\dots,\vg_N$ are 
all independent. Then the concatenated code $C^*=\cout\circ(\cin^1,\dots,\cin^N)$
is a $\left(1-R-\eps,q^{O(1/\eps^2)}\right)_{led}$-list
decodable code with probability at least $1-q^{-\Omega(nN)}$ over the choices
of $\cout,\vg_1,\dots,\vg_N$. Further, with high probability, $C^*$ has rate
$R$.
\end{theorem}


\begin{proof}
Let $q\ge 2$ and $R^{*}=R$ be the rate of the outer code (the inner
codes are chosen so that their dimension $k=n$, and therefore have rate $1$).


We define a segment of a codeword in $C^*$ as a sequence of consecutive
$q$-ary symbols generated by one particular inner code.  An assumption that
we will make for the ease of analysis (and which we will remove later) is that
erasures, which occur with relative rate $\rho$, will be equally distributed
among the concatenated codeword segments.  This means that in our received
word $\vy$, the result of each of the $N$ inner code encodings will contain
at most $\rho n$ erasures.

We will show that there exists some integer $L$ such that any subset
of $L+1$ distinct encoded messages has the property that they all match the
non-erased segments of the received word with low probability.  Then we'll
apply the union bound to show that with high probability, the code meets the
list decoding capacity for erasures.

Define $Q=q^k$ and $\p=1-R^*-\eps$.
Let $J=\lfloor \log_Q(L+1)\rfloor$.  Then there exists a subset of size at
least $J$ of our list (which is of size $L+1$) such that the set of
messages $\{\vm_1,\vm_2,...\vm_J\}$ will be linearly independent
over $\mathbb{F}_{Q}$.  This is because there are only $Q^J$ unique ways
to form linear sums of these messages over $\F_Q$.

Because of this fact and because $C_{out}$ is a random linear code, the set
$\{C_{out}(\vm_1),C_{out}(\vm_2),\cdots ,C_{out}(\vm_J)\}$ can be treated as a set
of independently chosen random vectors in $\mathbb{F}_{q^k}^N$.

Fix an $s$ so that $1\leq s\leq N$ and let $y_s$ represent a particular segment
of our received word.  (There are $N$ such segments over $\F_Q$).  In our list of $J$ outer
encoded messages, we denote by $i$ the size of the subset of these where for each
outer encoded message $C_{out}(\vm_t)$, restricted to the segment $s$, $C_{out}(\vm_t)$
is the zero vector.  $J-i$ is then the number of messages such that $C_{out}(\vm_t)$ is
not the zero vector when restricted to the segment $s$.

We can bound the probability that each of these messages match the received
word at this segment, in the unerased positions, as follows:

\begin{equation}\label{prob_match}
\pr[(C^*(\vm_t))_s \simeq y_s ]\le \left(\frac{1}{q^n}\right)^i\left(1-\frac{1}{q^n}\right)^{J-i}\cdot q^{-(1-\rho)n(J-i)}.
\end{equation}

In the above, the relationship $(C^*(\vm_t))_s \simeq y_s$ means that the concatenated
code, on message $\vm_t$, restricted to segment $s$, matches the received word
$y$ on segment $s$ at all unerased positions.

If $(C_{out}(\vm_t))_s = 0$, then we just assume that $(C^*(\vm_t))_s \simeq y_s$,
so this is an upper bound, and not an equality.

The first term in the RHS of \eqref{prob_match} is the probability that $i$ messages
at this segment map to the zero vector, and the second term is the probability
that $J-i$ messages map to something other than the zero vector.

The third term is the probability that those nonzero $J-i$ messages
match the received word in every unerased position.

Now since 

$$\left(1-\frac{1}{q^n}\right)^{J-i} \leq 1,$$

we have that

$$\pr[(C^*(\vm_t))_s \simeq y_s ]\le q^{-(1-\rho)nJ}\cdot q^{i(1-\rho)n-in}.$$

Also, because $(1-\rho)$ is always less than 1,

$$q^{i(1-\rho)n-in}\le 1.$$

Therefore

$$\pr[(C^*(\vm_t))_s \simeq y_s ]\le q^{-(1-\rho)nJ}.$$

The probability, then, that every message in the list matches the
received word in the unerased positions for a single segment, taken
over all possible choices of locations and sizes of $i$ is then
(by the union bound over such locations and sizes, noting that there
are at most $q^J$ ways to make these choices):

\begin{equation}\label{collected_rho}
\pr\left[\bigwedge_{t=1}^J (C^*(\vm_t))_s \simeq y_s\right]\le q^J\cdot q^{-(1-\rho)nJ}.
\end{equation}

Recalling that each inner code is chosen independently, the probability
that this is true for all segments is then

$$\pr\left[\bigwedge_{t=1}^J C^*(\vm_t) \simeq \vy\right]\le q^{JN}\cdot q^{-(1-\rho)nJN}.$$

Taking the union bound over all possible received words and lists of size $J$:

\begin{equation}\label{final_prob}
\pr[C^*~\text{is not}~(\rho,L)_{led}]\le q^{nN}\cdot q^{(1-\rho)nN}\cdot q^{kKJ}\cdot q^{JN}\cdot q^{-(1-\rho)nJN}.
\end{equation}

The first term in RHS of \eqref{final_prob} is an upper bound on
the number of possibilities for the erasure positions.
The second term is the number of ways to specify the
unerased positions, the third term is the number of possible lists of
size $J$, and the fourth and fifth terms come from the previous inequality.

Since $kK=R^*nN$, and $2>1+(1-\rho)$, this can be rewritten and simplified as:

$$\pr[C^*~\text{is not}~(\rho,L)_{led}]\le q^{-nNJ\left(\frac{-2}{J}-R^*-\frac{1}{n}+\left(1-\rho\right)\right)}.$$

If we can choose $n$, $R^*$, and $J$ appropriately so that:

$$\frac{-2}{J} - R^* -\frac{1}{n} + (1-\rho) \ge \epsilon/2,$$

then this probability will be exponentially small.

Setting
$n\ge J, J=\lceil \frac{6}{\epsilon}\rceil$
works.

We still need to fix the assumption that the $\rho$ fraction of erasures are
all distributed equally among the $N$ encoded segments.

Note that if we describe the fraction of erasures in each segment by $\rho_s$,
then

$$\sum_{s=1}^N \rho_s n = \rho nN.$$

The per-segment probability then becomes

$$\pr[(C^*(\vm_t))_s \simeq y_s]\le q^J\cdot q^{-(1-\rho_s)nJ}$$

and the probability for the entire received word becomes

$$\pr\left[\bigwedge_{t=1}^{J} C^*(\vm_t) \simeq y\right]\le \prod_{s=1}^N q^J\cdot q^{-(1-\rho_s)nJ}.$$

Note further that the $\rho_s$ terms can be collected in the exponent
and simplified to inequality \eqref{collected_rho}.

Finally, the claim that $C^*$ has rate $R$ follows from a similar argument to that
from~\cite{atri-venkat-soda-08} and
is omitted.
\end{proof}
\subsection*{Acknowledgments} We thank Venkat Guruswami and Parikshit Gopalan for 
helpful discussions. Thanks again to Madhu Sudan for kindly allowing us to include Lemma~\ref{lem:inverse-markov}
in this paper.

\bibliographystyle{abbrv}
\bibliography{two-thm-ld-full-02}








\end{document}